\documentclass[final,5p,times,twocolumn]{elsarticle}
\usepackage{paralist}
\usepackage{graphicx}
\usepackage{cite,array}
\usepackage{dsfont, bm, color}
\usepackage{mathrsfs}
\usepackage{amssymb}
\usepackage[cmex10]{amsmath}
\biboptions{numbers, sort&compress}

\newtheorem{theorem}{Theorem}

\newdefinition{definition}{Definition}
\newdefinition{assumption}{Assumption}
\newdefinition{example}{Example}
\newdefinition{remark}{Remark}
\newproof{proof}{Proof}

\DeclareMathOperator{\Id}{Id}
\DeclareMathOperator*{\esssup}{ess.\,sup}

\DeclareMathOperator{\sgn}{sgn}

\journal{}

\begin{document}

\begin{frontmatter}

\title{Symbolic Abstractions for Nonlinear Control Systems via Feedback Refinement Relation  \tnoteref{label0}}
\tnotetext[label0]{This work was supported by the H2020 ERC Starting Grant BUCOPHSYS, the EU H2020 Co4Robots Project, the Swedish Foundation for Strategic Research (SSF), the Swedish Research Council (VR) and the Knut och Alice Wallenberg Foundation (KAW).}

\author{Wei Ren \corref{cor1}}\ead{weire@kth.se}
\author{Dimos V. Dimarogonas}\ead{dimos@kth.se}

\address{Division of Decision and Control Systems, EECS, KTH Royal Institute of Technology, SE-10044, Stockholm, Sweden.}
\cortext[cor1]{Corresponding author.}

\begin{abstract}
This paper studies the construction of symbolic abstractions for nonlinear control systems via feedback refinement relation. Both the delay-free and time-delay cases are addressed. For the delay-free case, to reduce the computational complexity, we propose a new approximation approach for the state and input sets based on a static quantizer, and then a novel symbolic model is constructed such that the original system and the symbolic model satisfy the feedback refinement relation. For the time-delay case, both static and dynamic quantizers are combined to approximate the state and input sets. This leads to a novel dynamic symbolic model for time-delay control systems, and a feedback refinement relation is established between the original system and the symbolic model. Finally, a numerical example is presented to illustrate the obtained results.
\end{abstract}

\begin{keyword}
Symbolic abstraction, nonlinear control systems, quantizers, time-delay control systems, feedback refinement relation.
\end{keyword}

\end{frontmatter}

\section{Introduction}
\label{sec-intro}

To deal with controller synthesis problems efficiently and to guarantee the correct-by-design controller synthesis, symbolic models have attracted considerable attention in recent years; see \citep{Tabuada2006linear, Tabuada2009verification}. A symbolic model is a discrete approximation of a continuous system such that the controller designed for the discrete approximation can be refined to the controller for the original system \citep{Majumdar2012approximately}. Because of symbolic models, when complex specifications are considered, algorithmic machineries for controller synthesis of discrete systems can be used to synthesize controllers for continuous systems automatically \citep{Girard2007approximation}. As a result, the symbolic model based approach provides an alternative for controller synthesis of control systems.

The essence to this approach is to find an equivalence relation on the state space of dynamic systems. Such equivalence relation leads to a new system, which is on the quotient space and shares the properties of interest with the original system. In the literature, several abstraction techniques have been developed; see \citep{Pola2008approximately, Reissig2014feedback}. The commonly-used equivalence relation is the (bi)simulation relation and its approximate variants \citep{Pola2008approximately, Girard2010approximately, Girard2007approximation}, which capture equivalences of dynamic systems in the exact or approximate settings. However, these equivalence relations require a preassumption: \textit{the original system is incrementally stable or incrementally forward complete}, which constrains the class of control systems that can be studied via the symbolic model based approach. In addition, the exact information of the original system is needed to obtain the refined controller, which results in a huge computational complexity for the abstract controller due to its abstraction refinement. As a result, a feedback refinement relation was proposed in \citep{Reissig2014feedback} as an alternative equivalence relation, which connects the abstract controller to the original system via a static quantizer \citep{Reissig2017feedback}. Some results can be found on these equivalence relations; see \citep{Pola2008approximately, Girard2010approximately, Girard2007approximation} for the approximate (bi)simulation relations and \citep{Reissig2017feedback, Meyer2018compositional} for the feedback refinement relation. In previous works, these equivalence relations are based on uniform quantization, which partitions the space with a uniform distance \citep{Delchamps1990stabilizing}.

In this paper, we focus on symbolic abstractions of nonlinear control systems via the feedback refinement relation. Using both static and dynamic quantization, symbolic models are constructed for nonlinear control systems in both the delay-free and time-delay cases. Furthermore, the feedback refinement relation is ensured between the original system and the constructed symbolic model. Our results improve the existing techniques in two directions: i) by introducing a different static quantization and combining both static and dynamic quantization in the construction of symbolic abstractions, and ii) by extending the feedback refinement relation from the delay-free case to the time-delay case. In the first direction, our technique improves upon the results of \citep{Girard2010approximately, Reissig2017feedback} by applying a non-uniform quantizer in the construction of symbolic models for delay-free systems, and the results of \citep{Pola2010symbolic} by combining static and dynamic quantizers in the construction of symbolic models for time-delay systems. For the delay-free case, a logarithmic quantizer \citep{Liu2012sector, Fu2005sector}, which is static but partitions the space with different distances, is applied in the construction of symbolic models. For the time-delay case, the combination of the logarithmic and zoom quantizers, in which the zoom quantizer is a dynamic quantizer as in \citep{Liberzon2003hybrid, Ren2018quantized}, is implemented in the approximation of the state and input sets. In that sense, we propose an alternative approximation approach for the state and input sets. Due to the combination of the logarithmic and zoom quantizers, the computational complexity is reduced for the construction of symbolic models. In the second direction, our technique improves the results in \citep{Reissig2017feedback} by extending the feedback refinement relation from delay-free systems to time-delay systems. Besides the combination of the logarithmic and zoom quantizers, the spline functions are also involved in the construction of the symbolic model, and the feedback refinement relation is guaranteed from the original system to the symbolic model.

A preliminary version of this work has been presented in \citep{Ren2019logarithmic} where only the delay-free case is considered. The current paper refines and extends the approach to consider the time-delay case by combining static and dynamic quantization. Thus, the results in \citep{Ren2019logarithmic} are special cases of this paper.

The remainder of this paper is organized as follows. In Section \ref{sec-nonconsys}, both nonlinear control systems and time-delay control systems are introduced. Both static and dynamic quantizers are given in Section \ref{sec-quantizers}. Symbolic models are constructed for nonlinear control systems in Section \ref{sec-delayfree}, and for time-delay control systems in Section \ref{sec-timedelay}. Finally, a numerical example is presented in Section \ref{sec-example} to illustrate the developed results. Conclusions and future works are presented in Section \ref{sec-conclusion}.

\section{Nonlinear Control Systems}
\label{sec-nonconsys}

In this section, basic definitions of nonlinear control systems and time-delay control systems are introduced. To begin with, some necessary notations are introduced.

\subsection{Notations}
\label{subsec-note}

$\mathbb{R}:=(-\infty, +\infty)$; $\mathbb{R}^{+}_{0}:=[0, +\infty)$; $\mathbb{R}^{+}:=(0, +\infty)$; $\mathbb{N}:=\{0, 1, \ldots\}$; $\mathbb{N}^{+}:=\{1, 2, \ldots\}$. Given a vector $x\in\mathbb{R}^{n}$, $x_{i}$ denotes the $i$-th element of $x$; $|x_{i}|$ is the absolute value of $x_{i}$; $|x|:=(|x_{1}|, \ldots, |x_{n}|)\in\mathbb{R}^{n}$; $\|x\|:=\max\{|x_{1}|, \ldots, |x_{n}|\}$ denotes the infinity norm of $x\in\mathbb{R}^{n}$. Given $a, b\in(\mathbb{R}\cup\{\pm\infty\})^{n}$, we define the relations $<, >, \leq, \geq$ on $a, b$ component-wise; a cell $\llbracket a, b\rrbracket$ is the closed set $\{x\in\mathbb{R}^{n}|a_{i}\leq x_{i}\leq b_{i}\}$. $E\in\mathbb{R}^{n}$ is the vector whose components are 1. Given a measurable and essentially bounded function $f: \mathbb{R}^{+}_{0}\rightarrow\mathbb{R}^{n}$, the (essential) supremum norm of $f$ is $\|f\|_{\infty}:=\esssup\{\|f(t)\|, t\geq0\}$. For a given $\tau\in\mathbb{R}^{+}$, define $f_{\tau}$ such that $f_{\tau}(t):=f(t)$ for any $t\in[0, \tau]$, and $f(t)=0$ elsewhere; $f$ is locally essentially bounded if $f_{\tau}$ is essentially bounded.

Given $a, b\in\mathbb{R}$ with $a<b$, $\mathcal{C}([a, b], \mathbb{R}^{n})$ denotes the space of the continuous functions $f: [a, b]\rightarrow\mathbb{R}^{n}$. For any continuous function $x(s)$ defined on $-\Theta\leq s<a$ with $\Theta, a>0$, and given any fixed $t$ with $0\leq t<a$, the symbol $x_{t}$ denotes the element of $\mathcal{C}([a, b], \mathbb{R}^{n})$ defined by $x_{t}(\theta):=x(t+\theta)$ with $\theta\in[-\Theta, 0]$. Given two sets $A$ and $B$, if $A$ is a subset of $B$, then we denote by $\Id_{A}: A\hookrightarrow B$ or simply by $\Id$ the natural inclusion map taking any $a\in A$ to $\Id(a)\in B$. A relation $\mathcal{R}\subset A\times B$ is the map $\mathcal{R}: A\rightarrow2^{B}$ defined by $b\in\mathcal{R}(a)$ if and only if $(a, b)\in\mathcal{R}$. Given a relation $\mathcal{R}\subset A\times B$, $\mathcal{R}^{-1}$ denotes the inverse relation of $\mathcal{R}$, i.e., $\mathcal{R}^{-1}:=\{(b, a)\in B\times A: (a, b)\in\mathcal{R}\}$.

\subsection{Nonlinear Control Systems}
\label{subsec-nonsystem}

In this paper, we consider nonlinear control systems in both the delay-free and time-delay cases. For the delay-free case, the class of nonlinear control systems is formalized below.

\begin{definition}[\citep{Pola2008approximately}]
\label{def-1}
A \textit{control system} is a quadruple $\Sigma=(X, U, \mathcal{U}, f)$, where, (i) $X\subseteq\mathbb{R}^{n}$ is the state set; (ii) $U\subseteq\mathbb{R}^{m}$ is the input set; (iii) $\mathcal{U}$ is a subset of all piecewise continuous functions of time from the interval $(a, b)\subset\mathbb{R}$ to $U$ with $a<0<b$; (iv) $f: X\times U\rightarrow X$ is a continuous map satisfying the following Lipschitz assumption: given a bounded set $K\subseteq X\times U$, there exists a constant $L_{1}\in\mathbb{R}^{+}$ (which may be related to the set $K$), such that $\|f(x, u)-f(y, u)\|\leq L_{1}\|x-y\|$ for all $(x, u), (y, u)\in K$.
\end{definition}

A curve $\xi: (a, b)\rightarrow X$ is said to be a \textit{trajectory} of $\Sigma$, if there exists $\mathbf{u}\in\mathcal{U}$ such that $\dot{\xi}(t)=f(\xi(t), \mathbf{u}(t))$ for almost all $t\in(a, b)$. Different from the trajectory defined over the open domain, we refer to the trajectory $\mathbf{x}: [0, \tau]\rightarrow X$ defined on a closed interval $[0, \tau]$ with $\tau\in\mathbb{R}^{+}$ such that $\mathbf{x}=\xi|_{[0, \tau]}$. Denote by $\mathbf{x}(t, x, \mathbf{u})$ the point reached at time $t\in(a, b)$ under the input $\mathbf{u}\in\mathcal{U}$ from the initial condition $x\in X$. Such a point is uniquely determined, since the assumptions on $f$ ensure the existence and uniqueness of the state trajectory; see also \citep[Appendix C.3]{Sontag2013mathematical}. In addition, if $f(0, 0)=0$, then $x(t)=0$ is the trivial solution for the unforced system $\dot{x}(t)=f(x, 0)$. A system $\Sigma$ is said to be \emph{forward complete}, if every trajectory is defined on an interval of the form $(a, +\infty)$. Sufficient and necessary conditions can be found in \citep{Angeli1999forward} for the forward completeness of control systems.

\subsection{Time-delay Control Systems}
\label{subsec-delaysys}

Consider the following nonlinear time-delay control system
\begin{equation}
\label{eqn-1}
\left\{\begin{aligned}
\dot{x}(t)&=f(x_{t}, u(t-r)), \quad t\in\mathbb{R}^{+}, \\
x(t)&=\xi_{0}(t), \quad t\in[-\Theta, 0],
\end{aligned}\right.
\end{equation}
where $\Theta\in\mathbb{R}^{+}_{0}$ is the upper bound of the time delay that the system state involves, $r\in\mathbb{R}^{+}_{0}$ is the constant time delay that the input involves, $x(t)\in X\subseteq\mathbb{R}^{n}$ is the system state, $\xi_{0}\in\mathcal{C}^{0}([-\Theta, 0], X)$ is the initial condition, $x_{t}\in\mathcal{X}\subseteq\mathcal{C}^{0}([-\Theta, 0],  X)$ is the time-delay state, and $u(t)\in U\subseteq\mathbb{R}^{m}$ is the control input with the domain $[-r, +\infty)$. In the following, we formalize the definition for time-delay control systems.

\begin{definition}
\label{def-2}
A \textit{time-delay control system} is a sextuple $\bar{\Sigma}=(X, \mathcal{X}, \xi_{0}, U, \mathcal{U}, f)$, where, (i) $ X\subseteq\mathbb{R}^{n}$ is the state set; (ii) $\mathcal{X}\subseteq\mathcal{C}^{0}([-\Theta, 0], X)$ is the set of the time-delay states; (iii) $\xi_{0}\in\mathcal{X}$ is the initial condition; (iv) $U\subseteq\mathbb{R}^{m}$ is the input set; (v) $\mathcal{U}$ is a subset of all measurable and locally essentially bounded functions of time from $[-r, +\infty)$ to $U$; (vi) $f: \mathcal{X}\times U\rightarrow X$ is a continuous map satisfying the following Lipschitz assumption: given a bounded set $\bar{K}\subseteq\mathcal{X}\times U$, there exists a constant $L_{2}\in\mathbb{R}^{+}$ (which may be related to the set $\bar{K}$), such that $\|f(x, u)-f(y, u)\|\leq L_{2}\|x-y\|$ for all $(x, u), (y, u)\in\bar{K}$.
\end{definition}

Without loss of generality, we assume $f(0, 0)=0$, thus ensuring that $x(t)=0$ is the trivial solution for the unforced system $\dot{x}(t)=f(x_{t}, 0)$. Note that multiple discrete non-commensurate delays and distributed delays are allowed but bounded in \eqref{eqn-1}; see also \citep{Pola2010symbolic, Pola2015symbolic}. The assumptions on $f$ ensure the existence and uniqueness of the solution of \eqref{eqn-1}. Due to the time delays in \eqref{eqn-1}, both $\mathcal{X}$ and $\mathcal{U}$ are functional spaces, which is the main difficulty in the approximation of the state and input sets; see Section \ref{sec-timedelay}.

\subsection{Feedback Refinement Relation}
\label{subsec-feedrefine}

The objective of this paper is to construct symbolic models for nonlinear control systems in both the delay-free and time-delay cases. To this end, we next present a feedback refinement relation, which is a criterion to verify the relation between two transition systems. To begin with, the class of transition systems that will be used to denote symbolic models for control systems is introduced.

\begin{definition}[\citep{Tabuada2009verification}]
\label{def-3}
A \textit{transition system} is a sextuple $T=(X, X^{0}, U, \Delta, Y, H)$, where, (i) $X$ is the set of states; (ii) $X^{0}$ is the set of initial states; (iii) $U$ is the set of inputs; (iv) $\Delta\subseteq X\times U\times X$ is the transition relation; (v) $Y$ is the output set; (vi) $H : X\rightarrow Y$ is the output map.
\end{definition}

\begin{definition}[\citep{Reissig2017feedback}]
\label{def-4}
Let $T_{1}$ and $T_{2}$ be two transition systems with $T_{i}=(X_{i}, X^{0}_{i}, U_{i}, \Delta_{i}, Y_{i}, H_{i})$ for $i\in\{1, 2\}$, and assume that $U_{2}\subseteq U_{1}$. A relation $\mathcal{F}\subseteq X_{1}\times X_{2}$ is a \textit{feedback refinement relation} from $T_{1}$ to $T_{2}$, if for all $(x_{1}, x_{2})\in\mathcal{F}$, (i) $U_{2}(x_{2})\subseteq U_{1}(x_{1})$;
(ii) $u\in U_{2}(x_{2})\Rightarrow\mathcal{F}(\Delta_{1}(x_{1}, u))\subseteq\Delta_{2}(x_{2}, u)$, where $U_{i}(x):=\{u\in U_{i}: \Delta_{i}(x, u)\neq\varnothing\}$ and $i\in\{1, 2\}$. We denote by $T_{1}\preceq_{\mathcal{F}}T_{2}$ if $\mathcal{F}\subseteq X_{1}\times X_{2}$ is a feedback refinement relation from $T_{1}$ to $T_{2}$.
\end{definition}

\section{Static and Dynamic Quantizers}
\label{sec-quantizers}

To approximate the state and input sets, a general approach is based on quantization, which divides the state and input sets via a sequence of embedded lattices, whose intersection points are used to denote the abstract states and inputs in the local regions. In terms of quantization mechanisms, there are generally two types of quantizers: static quantizers and dynamic quantizers, which are introduced in the following subsections.

\subsection{Static Quantizer}
\label{subsec-staticquan}

A static quantizer is a memoryless and time-invariant piecewise-constant function $\mathbf{q}: \mathbb{R}^{n}\rightarrow\mathcal{Q}$, where $\mathcal{Q}$ is a finite subset of $\mathbb{R}^{n}$; see also \citep{Liberzon2003hybrid}. That is, the quantizer $\mathbf{q}$ divides $\mathbb{R}^{n}$ into a finite number of quantization regions of the form $\{z\in\mathbb{R}^{n}: \mathbf{q}(z)=\jmath\in\mathcal{Q}\}$. Since the quantization regions of static quantizers are time-invariant, static quantizers provide simple structures for the approximation of the state and input sets \citep{Ren2018quantized}. A commonly-used static quantizer is the uniform quantizer, which partitions the state and input sets uniformly \citep{Pola2008approximately}. However, to reduce the computational complexity, we introduce the following logarithmic quantizer, which provides an alternative for the approximation of the state and input sets.

\begin{definition}[\citep{Fu2005sector}]
\label{def-5}
A quantizer is called a \textit{logarithmic quantizer}, if it has the following form
\begin{align}
\label{eqn-2}
Q_{1}(z):=\left\{\begin{aligned}
&z_{i},  & & (1+\eta)^{-1}z_{i}<z\leq(1-\eta)^{-1}z_{i}; \\
&0, & & 0\leq z\leq(1+\eta)^{-1}d;  \\
& -Q_{1}(-z), & & z<0, \\
\end{aligned}\right.
\end{align}
where $z_{i}=\rho^{(1-i)}d$, $\rho=\frac{1-\eta}{1+\eta}$, $\eta\in(0, 1)$, $d>0$, and $i\in\mathbb{N}^{+}$.
\end{definition}

In Definition \ref{def-5}, the constant $\rho\in(0, 1)$ is called the quantization density; the constant $z_{\min}:=(1+\eta)^{-1}d$ determines the size of the deadzone. From \eqref{eqn-2}, $z_{i+1}=\rho^{-1}z_{i}$ for all $i\in\mathbb{N}^{+}$. The quantization error, which is defined as $z-Q_{1}(z)$, can be written as $z-Q_{1}(z)=\Lambda(z)z$ with $\Lambda(z)\in[-\eta, \eta]$; see \citep{Fu2005sector, Liu2012sector}. The quantizer $Q_{1}(z)$ takes values from the set $\mathcal{Q}:=\{0, \pm z_{i}: i\in\mathbb{N}^{+}\}$. For each $q\in\mathcal{Q}$, its quantization region is denoted by
\begin{align}
\label{eqn-3}
\hat{q}:=\left\{\begin{aligned}
&\ [\alpha_{1}q, \alpha_{2}q], \quad &&\text{ if } q\neq0,  \\
&\ [-z_{\min}, z_{\min}], \quad &&\text{ if } q=0,
\end{aligned}\right.
\end{align}
where $\alpha_{1}:=\max\{0, \sgn(q)\}(1+\eta)^{-1}+\max\{0, -\sgn(q)\}(1-\eta)^{-1}\},$ and $\alpha_{2}:=\max\{0, \sgn(q)\}(1-\eta)^{-1}+\max\{0, -\sgn(q)\}(1+\eta)^{-1}\}$. The closed quantization regions in \eqref{eqn-3} imply that the signal $z\in\hat{z}_{i}\cap\hat{z}_{i+1}$ can be quantized as either $z_{i}$ or $z_{i+1}$, $i\in\mathbb{N}^{+}$. Different from the uniform quantizer applied in \citep{Pola2008approximately, Pola2010symbolic, Reissig2017feedback}, the logarithmic quantizer partitions the state set non-uniformly, and the quantization becomes coarser as the distance between the origin and the signal gets larger.

\subsection{Dynamic Quantizer}
\label{subsec-dynquan}

A dynamic quantizer is based on the static quantizer and has a time-varying parameter, which adjusts quantization levels dynamically; see \citep{Ren2018quantized, Liberzon2003hybrid}. Therefore, the quantization regions are dynamic, which provides more flexibility for control design. Next, the zoom quantizer, which is a commonly-used dynamic quantizer \citep{Liberzon2003hybrid}, is introduced.

For a static quantizer $\mathbf{q}$, assume that there exist constants $M>\Lambda>0$ and $\Lambda_{0}>0$ such that the following conditions are satisfied; see also \citep{Ren2018quantized, Liberzon2003hybrid}.
\begin{align}
\label{eqn-4}
\|z\|\leq M&\Rightarrow\|\mathbf{q}(z)-z\|\leq\Lambda, \\
\label{eqn-5}
\|z\|>M&\Rightarrow\|\mathbf{q}(z)\|>M-\Lambda, \\
\label{eqn-6}
\|z\|\leq\Lambda_{0}&\Rightarrow\mathbf{q}(z)\equiv0.
\end{align}
In \eqref{eqn-4}-\eqref{eqn-5}, $M$ is called the range of the quantizer, and $\Lambda$ is called the upper bound of the quantization error $\mathbf{q}(z)-z$. Condition \eqref{eqn-4} implies that the quantization error is bounded by $\Lambda$ if the signal does not saturate. Condition \eqref{eqn-5} provides an approach to detecting whether the signal saturates or not. $\Lambda_{0}>0$ in \eqref{eqn-6} is called the size of the deadzone, which implies that the signal is quantized to zero directly if the signal is small enough.

Following the above assumption, the zoom quantizer is given as follows. For any signal $z\in\mathbb{R}$ to be quantized,
\begin{align}
\label{eqn-7}
Q_{2}(z, \delta)&:=\delta\mathbf{q}(z/\delta) \nonumber \\
& \quad =\left\{\begin{aligned}
&M\Lambda\delta,  && z\geq(M+0.5)\Lambda\delta; \\
&k\Lambda\delta, && (k-0.5)\Lambda\delta\leq z<(k+0.5)\Lambda\delta; \\
&-M\Lambda\delta,  && z<-(M+0.5)\Lambda\delta,
\end{aligned}\right.
\end{align}
where $\delta\in\mathbb{R}^{+}$ is the time-varying quantization parameter, and $k\in\mathcal{M}:=\{-M, \ldots, M\}$. For the zoom quantizer \eqref{eqn-7}, the quantization range is $M\Lambda\delta$; the upper bound of the quantization error is $\Lambda\delta$; and the size of the deadzone is $\Lambda\delta$, all of which are time-varying due to the dependence on $\delta$. For each quantized measurement $q\in\{k\Lambda\delta: k\in\mathcal{M}\}$, the corresponding quantization region is given by $\hat{q}:=[(k-0.5)\Lambda\delta, (k+0.5)\Lambda\delta]$.

From geometrical considerations, we have that for all $z\in\mathbb{R}$, $\|z-Q_{2}(z, \delta)\|\leq\Lambda\delta$ if $\|z\|\leq(M+1)\Lambda\delta$; otherwise, $\|z-Q_{2}(z, \delta)\|>\Lambda\delta$. The quantization error is bounded in the region $\{z\in\mathbb{R}: \|z\|\leq(M+1)\Lambda\delta\}$. As a result, the zoom quantizer \eqref{eqn-7} is only available for bounded regions; see \citep{Ren2018quantized, Liberzon2003hybrid}.

\section{Symbolic Model for Delay-free Case}
\label{sec-delayfree}

In this section, we focus on the construction of the symbolic abstraction for nonlinear control systems in the delay-free case. To this end, we work with the time-discretization of the system $\Sigma$. Assume the sampling period is $\tau>0$, which is a design parameter. We define the time-discretization of the system $\Sigma$ as the transition system $T_{\tau}(\Sigma):=(X_{1}, X^{0}_{1}, U_{1}, \Delta_{1}, Y_{1}, H_{1})$, where,
\begin{itemize}
  \item the state set is $X_{1}:=X$;
  \item the set of initial states is $X^{0}_{1}:=X$;
  \item the input set is $U_{1}:=\{u\in\mathcal{U}: \mathbf{x}(\tau, x, u) \text{ is defined for all } x\in X\}$;
  \item the transition relation is given as follows: for $x\in X_{1}$ and $u\in U_{1}$, $x'=\Delta_{1}(x, u)$ if and only if $x'=\mathbf{x}(\tau, x, u)$;
  \item the output set is $Y_{1}:=X$;
  \item the output map is $H: X_{1}\hookrightarrow X_{1}$.
\end{itemize}

\subsection{Approximation of State and Input Sets}
\label{subsec-nonapprox}

Using the quantizer \eqref{eqn-2}, we approximate the state and input sets in this subsection. The states set $X$ is approximated by the sequence of embedded lattices $[X]_{\eta}:=\{q\in X: q_{i}\in\{0, \pm\rho^{(1-k_{i})}d\}, k_{i}\in\mathbb{N}^{+}, i\in\{1, \ldots, n\}\}$, where $\rho=(1+\eta)^{-1}(1-\eta)$, $\eta\in(0, 1)$ is treated as a state space parameter, and $d>0$ is a fixed constant. For each $q\in[X]_{\eta}$, its quantization region is given by $\hat{q}$ as in \eqref{eqn-3}. We associate a quantizer $Q_{\eta}: X\rightarrow[X]_{\eta}$ such that $Q_{\eta}(x)=Q_{1}(x)=q$ if and only if for $x=(x_{1}, \ldots, x_{n})\in X$ and $i\in\{1, \ldots, n\}$,
\begin{align*}
\frac{|q_{i}|}{1+\eta}\leq|x_{i}|\leq\frac{|q_{i}|}{1-\eta} \text{ or } \frac{-d}{1+\eta}\leq x_{i}\leq\frac{d}{1+\eta}.
\end{align*}
As a result, we obtain from \eqref{eqn-2} and simple geometrical considerations that for all $x\in X$, $\|x-Q_{\eta}(x)\|\leq\Lambda(x)\|x\|$, where $\Lambda(x)\in[-\eta, \eta]$. With the quantizer $Q_{\eta}$, the state set $X$ is partitioned as $\hat{X}$ with
\begin{align}
\label{eqn-8}
\hat{X}&:=\bigcup_{q\in[\mathbb{R}^{n}]_{\eta}}\hat{q}\cap X,
\end{align}
where $\hat{q}$ is the quantization region corresponding to the quantized measurement $q\in[\mathbb{R}^{n}]_{\eta}$. Note that the set $\hat{X}$ includes the set $X\backslash(\bigcup_{q\in[X]_{\eta}}\hat{q})$. If $X=\mathbb{R}^{n}$, we can define $\hat{X}:=\bigcup_{q\in[\mathbb{R}^{n}]_{\eta}}\hat{q}$.

In the following, the approximation of the  input set $U_{1}$ is presented. We approximate $U_{1}$ by means of the set
\begin{equation}
\label{eqn-9}
U_{2}:=\bigcup_{\tilde{q}\in\hat{X}}U_{2}(\tilde{q}),
\end{equation}
where $U_{2}(\tilde{q})$ captures the set of inputs applied at the state $\tilde{q}\in\hat{X}$. Here, $\tilde{q}$ plays the same role as the quantization region $\hat{q}$. If $\hat{q}\in\hat{X}$, then $\tilde{q}=\hat{q}$; otherwise, $\tilde{q}=\hat{q}\bigcap X$. In addition, the quantized measurement $q\in[\mathbb{R}^{n}]_{\eta}$ corresponding to $\hat{q}$ is treated as the quantized measurement corresponding to $\tilde{q}$.

\subsection{Symbolic Model}
\label{sec-nondelaysymbolic}

With the approximation of the state and input sets, the symbolic abstraction of the system $T_{\tau}(\Sigma)$ is developed in this subsection. The developed symbolic abstraction is a transition system $T_{\tau, \eta}(\Sigma)=(X_{2}, X^{0}_{2}, U_{2}, \Delta_{2}, Y_{2}, H_{2})$, where
\begin{itemize}
\item the set of states is $X_{2}=\hat{X}$, which is given in \eqref{eqn-8};
\item the set of initial states is $X^{0}_{2}=\hat{X}$;
\item the set of inputs is $U_{2}=\bigcup_{\tilde{q}\in\hat{X}}U_{2}(\tilde{q})$;
\item the transition relation is given as follows: for $\tilde{q}_{1}, \tilde{q}_{2}\in X_{2}$ and $u\in U_{2}$, $\tilde{q}_{2}\in\Delta_{2}(\tilde{q}_{1}, u)$ if and only if
\begin{align}
\label{eqn-10}
\tilde{q}_{2}&\bigcap\left(\mathbf{x}(\tau, q_{1}, u)+\left\llbracket-\theta_{1}e^{L_{1}\tau}\bar{q}_{1}, \theta_{1}e^{L_{1}\tau}\bar{q}_{1}\right\rrbracket\right)\neq\varnothing,
\end{align}
where $q_{1}$ is the quantized measurement corresponds to $\tilde{q}_{1}$; $\theta_{1}:=\eta(1-\eta)^{-1}$; $\bar{q}_{1}:=|q_{1}|+E_{q_{1}}$; $E_{q_{1}}\in\mathbb{R}^{n}$ is a vector whose components are 1 if the corresponding components of $q_{1}$ are 0, and 0 otherwise; and $L_{1}>0$ is the Lipschitz constant of $f$ in $\tilde{q}_{1}\in X_{2}$;
\item the set of outputs is $Y_{2}=X$;
\item the output map is $H_{2}=\Id_{X_{2}}$.
\end{itemize}

In the construction of the symbolic abstraction $T_{\tau, \eta}(\Sigma)$, the technique applied in \eqref{eqn-10} is similar to those in \citep{Kim2017symbolic, Meyer2018compositional, Reissig2017feedback}, where the overapproximation of successors of states is applied. In \eqref{eqn-10}, $L_{1}$ is the Lipschitz constant for the current abstract state $\tilde{q}_{1}\in\hat{X}$ and can be computed accordingly, and further used to determine the next abstract state $\tilde{q}_{2}$. In addition, $\theta_{1}e^{L_{1}\tau}\bar{q}_{1}$ plays the same role as the growth bound in \citep{Reissig2017feedback}, and the transition relation \eqref{eqn-10} is similar to that for the sparse abstraction in \citep{Kim2017symbolic}. Since the logarithmic quantizer is implemented here, $\theta_{1}e^{L_{1}\tau}\bar{q}_{1}$ is related to the quantized measurements, which is different from \citep{Reissig2017feedback}. In particular, the quantized measurements are not the centers of the corresponding quantization regions; see Section \ref{sec-example}. Hence, the developed symbolic abstraction provides an alternative for the construction of symbolic abstractions with respect to earlier works \citep{Pola2008approximately, Reissig2017feedback}.

For the transition systems $T_{\tau}(\Sigma)$ and $T_{\tau, \eta}(\Sigma)$, the next theorem establishes the feedback refinement relation between them.

\begin{theorem}
\label{thm-1}
Consider the system $\Sigma$ with the time and state space sampling parameters $\tau, \eta\in\mathbb{R}^{+}$. Let  the map $\mathcal{F}: X_{1}\rightarrow X_{2}$ be given by $\mathcal{F}(x)=\tilde{q}$ if and only if $x\in\tilde{q}$. Then $T_{\tau}(\Sigma)\preceq_{\mathcal{F}}T_{\tau, \eta}(\Sigma)$.
\end{theorem}

\begin{proof}
Following from the definitions of the systems $T_{\tau}(\Sigma)$ and $T_{\tau, \eta}(\Sigma)$, we have that $U_{2}\subseteq U_{1}$. Let $(x_{1}, \tilde{q}_{1})\in\mathcal{F}$ with $x_{1}\in X_{1}$ and $\tilde{q}_{1}\in X_{2}$, then $x_{1}\in\tilde{q}_{1}$ holds from the feedback refinement relation. For each $u\in U_{2}(\tilde{q}_{1})$, we obtain that $u\in U_{2}(\tilde{q}_{1})\subseteq U_{2}\subseteq U_{1}$. Moreover, $\Delta_{2}(\tilde{q}_{1}, u)\neq\varnothing$ holds from the definition of $U_{2}(\tilde{q}_{1})$. If $\Delta_{1}(x_{1}, u)=\varnothing$, then we have that $u\notin U_{1}$, which is a contradiction. Hence, $\Delta_{1}(x_{1}, u)\neq\varnothing$, which implies that $u\in U_{1}(x_{1})$. We thus conclude that $U_{2}(\tilde{q}_{1})\subseteq U_{1}(x_{1})$.

Given $\tilde{q}_{1}, \tilde{q}_{2}\in X_{2}$ and $u\in U_{2}(\tilde{q}_{1})$, define $x_{2}:=\Delta_{1}(x_{1}, u)$. Since $(x_{1}, \tilde{q}_{1})\in\mathcal{F}$ holds, we have that $x_{1}\in\tilde{q}_{1}$, which implies from the logarithmic quantizer \eqref{eqn-2} that $\|x_{1}-q_{1}\|\leq\theta_{1}\|q_{1}\|$. If $\Delta_{1}(x_{1}, u)\cap\tilde{q}_{2}\neq\varnothing$, then there exists $\mathbf{x}(\tau, x_{1}, u)\in\tilde{q}_{2}$. Since $f$ in $\Sigma$ satisfies the Lipschitz condition, one has that $\|\mathbf{x}(\tau, x_{1}, u)-\mathbf{x}(\tau, q_{1}, u)\|\leq e^{L_{1}\tau}\|x_{1}-q_{1}\|\leq\theta_{1}e^{L_{1}\tau}\|q_{1}\|$, which further implies that $\tilde{q}_{2}\bigcap(\mathbf{x}(\tau, q_{1}, u)+\llbracket-\theta_{1}e^{L_{1}\tau}\bar{q}_{1}, \theta_{1}e^{L_{1}\tau}\bar{q}_{1}\rrbracket)\neq\varnothing$. We thus obtain from the construction of the abstraction $T_{\tau, \eta}(\Sigma)$ that $\tilde{q}_{2}\in X_{2}$, which in turn implies that $(x_{2}, \tilde{q}_{2})\in\mathcal{F}$.
\hfill$\blacksquare$
\end{proof}

In the proof of Theorem \ref{thm-1}, the intersection $\Delta_{1}(x_{1}, u)\cap\tilde{q}_{2}\neq\varnothing$ holds if the state set $X$ is unbounded. If $X$ is bounded as in practical systems, we have that $\Delta_{1}(x_{1}, u)\cap\tilde{q}_{2}=\varnothing$ holds in $\mathbb{R}^{n}\backslash X$. In this case, we can impose an additional requirement such that $\Delta_{2}(\tilde{q}_{1}, u)=\varnothing$ if $\tilde{q}_{1}\notin X_{2}$. Therefore, the feedback refinement relation is still valid in this case; see also \citep{Ren2019logarithmic, Reissig2017feedback, Meyer2018compositional}. In addition, since the approximation error has great effects on the accuracy of the developed symbolic model, we can reduce the approximation error by implementing static or dynamic quantization on the components of $X_{2}$, which will be discussed in the next section. For instance, using another static quantization on $\tilde{q}\in X_{2}$, smaller quantization regions are obtained and then the approximation accuracy is improved. Note that such setting does not need to rediscretize the state and input sets, and can be treated as an abstraction refinement strategy.

\section{Symbolic Model for Time-delay Case}
\label{sec-timedelay}

In this section, we study the construction of the symbolic model for time-delay control systems. In this case, assume that the controller is digital, i.e., the control inputs are piecewise-constant. In many practical applications, the controllers are implemented through digital devices, which results in digital control inputs; see \citep{Majumdar2012approximately}. In the following, we refer to the time-delay systems with digital controllers as digital time-delay control systems. Suppose that the input set $U$ contains the origin, and that the control inputs belong to the set $\mathcal{U}_{\tau}:=\{u\in\mathcal{U}: u: [-r, -r+\tau]\rightarrow U \text{ and } u(t)=u(-r), \forall t\in[-r, -r+\tau]\}$, where $\tau\in\mathbb{R}^{+}$ is the sampling period. Given the system $\bar{\Sigma}$, define the transition system $T_{\tau}(\bar{\Sigma}):=(\bar{X}_{1}, \bar{X}^{0}_{1}, \bar{U}_{1}, \bar{\Delta}_{1}, \bar{Y}_{1}, \bar{H}_{1})$ with
\begin{itemize}
  \item the state set $\bar{X}_{1}:=\mathcal{X}$;
  \item the set of initial states $\bar{X}^{0}_{1}:=\xi_{0}$, which is given in \eqref{eqn-1};
  \item the input set $\bar{U}_{1}:=\{u\in\mathcal{U}_{\tau}: \mathbf{x}(\tau, x, u) \text{ is defined for } x\in\mathcal{X}\}$;
  \item the transition relation given as follows: for $x\in\bar{X}_{1}$ and $u\in\bar{U}_{1}$, $x'=\bar{\Delta}_{1}(x, u)$ if and only if $x'=\mathbf{x}(\tau, x, u)$;
  \item the output set $\bar{Y}_{1}:=\mathcal{X}$;
  \item the output map $\bar{H}_{1}: \bar{X}_{1}\hookrightarrow\bar{X}_{1}$.
\end{itemize}

Similar to the delay-free case, the system $T_{\tau}(\bar{\Sigma})$ can be treated as a time discretization of the system $\bar{\Sigma}$. Note that $T_{\tau}(\bar{\Sigma})$ is not symbolic, since both $\bar{X}_{1}$ and $\mathcal{U}_{\tau}$ are functional spaces.

\subsection{Dynamic Spline-based Approximation}
\label{subsec-delayapprox}

To approximate the functional spaces $\mathcal{X}$ and $\mathcal{U}_{\tau}$, both logarithmic and zoom quantizers are applied in this subsection. The detailed approximation is explained below.

To begin with, we consider the delay-free version of the system $\bar{\Sigma}$, that is,
\begin{align}
\label{eqn-11}
\begin{aligned}
\dot{x}(t)&=f(x(t), u(t)), \quad t\in\mathbb{R}^{+}, x(0)=\xi_{0}(0),
\end{aligned}
\end{align}
where $x(t)\in X\subseteq\mathbb{R}^{n}$ is the system state, $u(t)\in\bar{\mathcal{U}}_{\tau}\subseteq\mathbb{R}^{m}$ is the control input, and $\bar{\mathcal{U}}_{\tau}$ denotes the set $\mathcal{U}_{\tau}$ with $r\equiv0$. From Subsection \ref{subsec-nonapprox}, we can obtain the approximation of the state set $X$ and the input set $\bar{\mathcal{U}}_{\tau}$. That is, there exists $\eta\in\mathbb{R}^{+}$ such that the state set $X$ is partitioned by the sequence of embedded lattices $[X]_{\eta}$. As a result, the state set $X$ and the input set $\bar{\mathcal{U}}_{\tau}$ are partitioned as $\hat{X}$ in \eqref{eqn-8} and $U_{2}$ in \eqref{eqn-9}, respectively.

\begin{figure}[!t]
\begin{center}
\begin{picture}(70, 85)
\put(-50,-19){\resizebox{60mm}{35mm}{\includegraphics[width=2.5in]{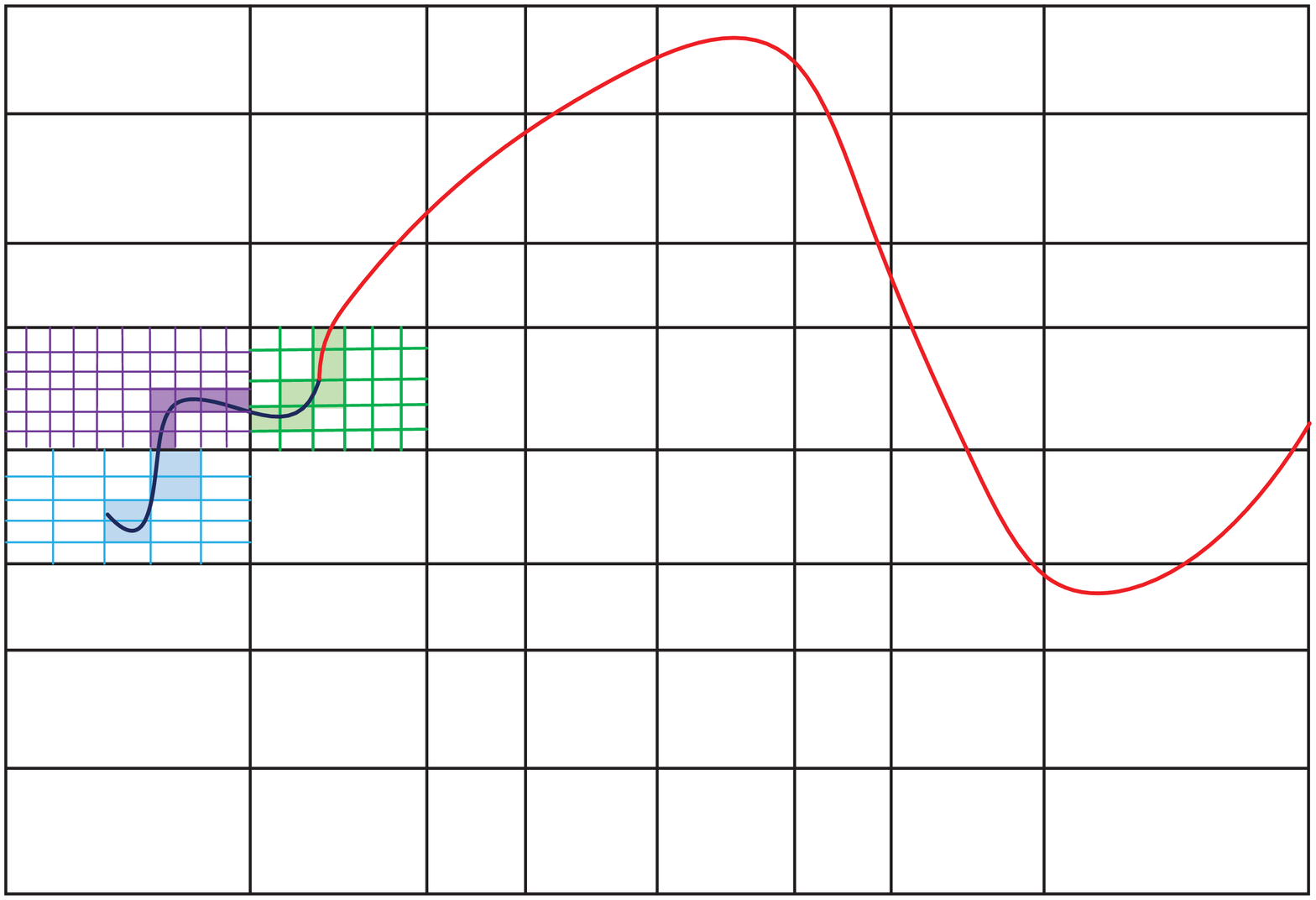}}}
\end{picture}
\end{center}
\caption{Illustration of the spline-based approximation scheme for a functional space. The black lines are obtained from logarithmic quantization. The time-delay initial condition (i.e., the black curve) is covered by 3 quantization regions. For each quantization region, a zoom quantizer is applied for further partition. The smaller quantization regions are thus used to approximate the time-delay initial condition; see the shadow regions.}
\label{fig-1}
\end{figure}

Next, based on the defined partition in \eqref{eqn-8}-\eqref{eqn-9}, we consider the system $\bar{\Sigma}$. For the time-delay initial state $\xi_{0}(t)$ with $t\in[-\Theta, 0]$, there exists a union of finite regions $\bigcup_{1\leq i\leq N_{0}}\tilde{q}_{i}\subseteq\hat{X}$ with $\tilde{q}_{i}\in\hat{X}$ and finite $N_{0}\in\mathbb{N}^{+}$ such that $\xi_{0}(t)\in\bigcup_{1\leq i\leq N_{0}}\tilde{q}_{i}$ for all $t\in[-\Theta, 0]$; see Figure \ref{fig-1}. To see this, assume that there exists a $t_{1}\in(-\Theta, 0]$ such that
\begin{equation}
\label{eqn-12}
\xi_{0}(t)\in\tilde{q}_{1}\in\hat{X}, \quad t\in[-\Theta, t_{1}].
\end{equation}
Now, for the bounded region $\tilde{q}_{1}\in\hat{X}$, the zoom quantizer \eqref{eqn-7} is applied to obtain the approximation of $\tilde{q}_{1}$. Assume that the quantization parameter $\delta=\delta_{1}\in\mathbb{R}^{+}_{0}$ in the region $\tilde{q}_{1}$. Therefore, the region $\tilde{q}_{1}$ is approximated via the following lattices
\begin{align}
\label{eqn-13}
[\tilde{q}_{1}]_{\delta_{1}}&:=\left\{q\in\tilde{q}_{1}: q_{i}=k_{i}\Lambda\delta_{1}, k_{i}\in\mathcal{M}, i\in\{1, \ldots, n\}\right\}.
\end{align}
That is, the region $\tilde{q}_{1}\in\hat{X}$ is partitioned into smaller regions, i.e., $\bigcup_{q\in[\tilde{q}_{1}]_{\delta_{1}}}\hat{q}\cap X$, where $\hat{q}$ is given in Subsection \ref{subsec-dynquan}. In addition, we can find a finite number of smaller regions to cover the trajectory of $\xi_{0}(t)$ in the interval $[-\Theta, t_{1}]$. Note that $\delta_{1}$ here is allowed to be zero, which is used to indicate that some regions do not need the zoom quantization. That is, given a $\tilde{q}\in\hat{X}$, $[\tilde{q}]_{0}=\tilde{q}$, which implies that we do not need to approximate $\tilde{q}$.

Similarly, we can find a $t_{2}\in(t_{1}, 0]$ such that $\xi_{0}(t)\in\tilde{q}_{2}\in\hat{X}$ for all $t\in[t_{1}, t_{2}]$. For the region $\tilde{q}_{2}\in\hat{X}$, the zoom quantizer \eqref{eqn-7} is applied with the quantization parameter $\delta=\delta_{2}\in\mathbb{R}^{+}_{0}$, and thus we approximate the region $\tilde{q}_{2}\in\hat{X}$ via the lattices $[\tilde{q}_{2}]_{\delta_{2}}$ with a similar form as in \eqref{eqn-13}. Note that $\delta_{1}$ and $\delta_{2}$ are not necessarily the same, and the choice of the quantization parameter in each region depends on the requirements of the desired approximation accuracy and the state set. For instance, if the approximation error is getting larger in the region $\tilde{q}_{1}$, then we can contract $\delta_{2}$ to improve the approximation accuracy in $\tilde{q}_{2}$; in case that obstacle avoidance is required, $\delta$ can be contracted to generate more admissible symbolic states. The approximation accuracy increases with the decrease of $\delta$.

Using the above technique iteratively, we approximate these components in $\hat{X}$ covering the initial functional $\xi_{0}\in\mathcal{C}^{0}([-\Theta, 0], X)$ by the lattices $\bigcup_{1\leq i\leq N_{0}}[\tilde{q}_{i}]_{\delta_{i}}$. Since the initial functional is known \textit{a priori} and the approximation $\hat{X}$ is obtained via the logarithmic quantizer, a finite  $N_{0}\in\mathbb{N}^{+}$ exists, and thus the above procedure ends after $N_{0}$ iterations.

For the evolution of the system state, we can use a similar mechanism and approximate the corresponding region in $\hat{X}$. As a result, combining both the logarithmic quantizer \eqref{eqn-2} and the zoom quantizer \eqref{eqn-7}, we define the set $\bar{X}:=\bigcup_{\tilde{q}_{i}\in\hat{X}}[\tilde{q}_{i}]_{\delta_{i}}$, and further partition the state set $X$ as
\begin{align}
\label{eqn-14}
\check{X}&:=\bigcup_{\tilde{q}_{i}\in\hat{X}}\bigcup_{q\in[\tilde{q}_{i}]_{\delta_{i}}}\hat{q}\cap X,
\end{align}
where $\tilde{q}_{i}\in\hat{X}$ is given in \eqref{eqn-8} and is determined by the logarithmic quantizer; $\delta_{i}\in\mathbb{R}^{+}_{0}$ is the quantization parameter for the zoom quantizer applied in the region $\tilde{q}_{i}\in\hat{X}$. Similarly, we can define the partition of the input set $\bar{U}_{1}$ as
\begin{equation}
\label{eqn-15}
\bar{U}_{2}:=\bigcup_{\tilde{q}\in\check{X}}\bar{U}_{2}(\tilde{q}), \quad i\in\mathbb{N}^{+},
\end{equation}
where $\bar{U}_{2}(\hat{q})$ captures the set of inputs that can be applied at the state $\tilde{q}\in\check{X}$ of the symbolic model.

\begin{remark}
\label{rmk-1}
This step is a combination of logarithmic and zoom quantization, and thus a finer approximation is obtained for the state and input sets than the delay-free case. In terms of the delay-free case, this step can be treated as a refinement for the approximation obtained in Subsection \ref{subsec-nonapprox}, and thus the approximation accuracy is improved via the combination of logarithmic and zoom quantization. To be specific, if the symbolic abstraction proposed in Subsection \ref{subsec-nonapprox} is coarse, then the combination of logarithmic and zoom quantization can be applied to refine the state set further to improve the approximation accuracy; see also \citep{Meyer2018compositional1}. In terms of the time-delay case, this setup lays a solid foundation for the following approximation of the functional spaces. In particular, the refined approximation of the state and input sets will be used to determine both the abstract state and the transition relation in the construction of symbolic abstractions for the time-delay case.
\hfill $\square$
\end{remark}

Finally, we use the partition $\check{X}$ to approximate the functional set $\mathcal{X}$. Given a constant $N\in\mathbb{N}^{+}$, consider the following spline functions (see also \citep{Schultz1973spline, Michel2013lectures}):
\begin{align*}
s_{0}(t)&:=\left\{\begin{aligned}
&1-(t-a)/h, &\ & t\in[a, a+h]; \\
&0, &\ &  \text{otherwise},
\end{aligned}\right. \\
s_{j}(t)&:=\left\{\begin{aligned}
&1-j+(t-a)/h, &\ & t\in[a+(j-1)/h, a+jh];\\
&1+j-(t-a)/h, &\ & t\in[a+jh, a+(j+1)/h]; \\
&0, &\ &  \text{otherwise},
\end{aligned}\right. \\
s_{N+1}(t)&:=\left\{\begin{aligned}
&1-(t-b)/h, &\ & t\in[b-h, b]; \\
&0, &\ &  \text{otherwise},
\end{aligned}\right.
\end{align*}
where $h:=(b-a)/(N+1)$. With these spline functions $s_{j}$, $j\in\{0, \ldots, N+1\}$, the set $\mathcal{X}$ can be approximated below.

Consider a set $\mathcal{X}\subseteq\mathcal{C}^{0}([a, b], X)$ with $X\subseteq\mathbb{R}^{n}$ and $[a, b]\subseteq\mathbb{R}$. Given a function $x\in\mathcal{X}$, we approximate $x\in\mathcal{X}$ by means of the union of the regions in $\check{X}$ covering $x\in\mathcal{X}$. Such union is obtained by the combination of $N+2$ spline functions $s_{j}$ centered at time $t=a+ih$ with the regions $\tilde{q}_{j}\in\check{X}$, which are defined as $\tilde{q}_{j}=\arg\{\tilde{q}\in\check{X}: x(a+jh)\in\tilde{q}\}$. Define an operator $\psi_{2}: \mathcal{X}\rightarrow\mathcal{C}^{0}([a, b], X)$ as follows:
\begin{equation}
\label{eqn-16}
\psi_{2}(x)(t):=\bigcup_{s_{j}(t)\neq0}\tilde{q}_{j}, \quad t\in[a, b],
\end{equation}
where $\tilde{q}_{j}\in\check{X}$, $x(a+jh)\in\tilde{q}_{j}$ and $j\in\{0, 1, \ldots, N+1\}$. Note that the operator $\psi_{2}$ is not uniquely defined. Let $a=-\Theta$ and $b=0$, and then the approximation of $\mathcal{X}$ is defined as $\bar{X}_{2}:=\psi_{2}(\mathcal{X})$.

In the final step, the spline functions are used to indicate all the regions in $\check{X}$ covering the function $x\in\mathcal{X}$. The number $N+2\in\mathbb{N}^{+}$ of the spline functions is based on the approximation accuracy and the applied zoom quantization. If the quantization parameter of the zoom quantizer is getting smaller, the number of the components in $\check{X}$ increases and thus a large $N$ leads to the improvement of the approximation accuracy due to the implementation of more components in $\check{X}$. Thus, we can write $N$ as $N(\bm{\delta})$ with $\bm{\delta}:=(\delta_{1}, \delta_{2}, \ldots)$. Once $N$ exceeds the threshold $M^{2}$ with $M$ given in \eqref{eqn-7}, the approximation accuracy is not affected because the following may occur: $\tilde{q}_{j}$ and $\tilde{q}_{k}$ are the same for some $j, k\in\{0, \ldots, N+1\}$ and $j\neq k$, which does not affect \eqref{eqn-16} and thus does not affect the approximation accuracy.

\subsection{Symbolic Model}
\label{subsec-delaysymbolic}

With the approximation of the state and input sets, we are ready to construct the symbolic model for the system $T_{\tau}(\bar{\Sigma})$ in this subsection. Given the parameters $\tau, \eta, \bm{\delta}$ as in Subsection \ref{subsec-delayapprox}, the constructed symbolic abstraction is a transition system $T_{\tau, \eta, \bm{\delta}}(\bar{\Sigma})=(\bar{X}_{2}, \bar{X}^{0}_{2}, \bar{U}_{2}, \bar{\Delta}_{2}, \bar{Y}_{2}, \bar{H}_{2})$, where,
\begin{itemize}
\item the set of states is $\bar{X}_{2}=\psi_{2}(\mathcal{X})$ with $\psi_{2}$ given in \eqref{eqn-16};
\item the set of initial states is $\bar{X}^{0}_{2}=\psi_{2}(\xi_{0})$ with $\xi_{0}$ given in \eqref{eqn-1};
\item the set of inputs is $\bar{U}_{2}$ given in \eqref{eqn-15};
\item the transition relation is given as follows: for $\tilde{q}_{1}, \tilde{q}_{2}\in\bar{X}_{2}$ and $u\in \bar{U}_{2}$, $\tilde{q}_{2}\in\bar{\Delta}_{2}(\tilde{q}_{1}, u)$ if and only if
\begin{align}
\label{eqn-17}
\tilde{q}_{2}&\cap\left(\mathbf{x}(\tau, q_{1}, u)+2\theta_{2}e^{L_{2}\tau}\llbracket-E, E\rrbracket\right)\neq\varnothing,
\end{align}
where $q_{1}:=\Sigma^{N+1}_{j=0}q_{1j}s_{j}$ with spline functions $s_{j}$ and quantized measurements $q_{1j}\in\bar{X}$ corresponding to $\tilde{q}_{1j}\in\tilde{q}_{1}$, $\theta_{2}:=\max_{s_{j}\neq0}\{\Lambda\delta_{j}\}$ with $\Lambda$ in \eqref{eqn-7} and the quantization parameters $\delta_{j}$ involved in $\tilde{q}_{1}$, and $L_{2}>0$ is the Lipschitz constant of $f$ in $\tilde{q}_{1}\in\bar{X}_{2}$;
\item the set of outputs is $\bar{Y}_{2}=\mathcal{X}$;
\item the output map is $\bar{H}_{2}=\Id_{\bar{X}_{2}}$.
\end{itemize}

In the construction of the symbolic model, both the spline functions $s_{j}$ and the regions $\tilde{q}_{1j}\in\check{X}$ covering the time-delay state are used in \eqref{eqn-17} to determine the transition relation. That is, for $\tilde{q}_{1}=\bigcup_{s_{j}(t)\neq0}\tilde{q}_{1j}$ with $t\in[-\Theta, 0]$ and $\tilde{q}_{1j}\in\check{X}$, $q_{1}:=\Sigma^{N+1}_{j=0}q_{1j}s_{j}$ is the quantized measurement corresponding to $\tilde{q}_{1}$ and used to approximate the time-delay state and to determine the next abstract state. Due to the combination of logarithmic and zoom quantizers, the growth bound $2\theta_{2}e^{L_{2}\tau}E$ in \eqref{eqn-17} is not related to the $q_{1}$, which is different from \eqref{eqn-10}. Similar to the delay-free case in Section \ref{sec-delayfree}, the following theorem establishes the feedback refinement relation from $T_{\tau}(\bar{\Sigma})$ to $T_{\tau, \eta, \bm{\delta}}(\bar{\Sigma})$.

\begin{theorem}
\label{thm-2}
Consider the system $\bar{\Sigma}$ with the time and state space sampling parameters $\tau, \eta\in\mathbb{R}^{+}, \delta_{i}\in\mathbb{R}^{+}_{0}$, $i\in\mathbb{N}^{+}$. Let  the map $\mathcal{F}: X_{1}\rightarrow X_{2}$ be given by $\mathcal{F}(x)=\tilde{q}$ if and only if $x\in\tilde{q}$. Then $T_{\tau}(\bar{\Sigma})\preceq_{\mathcal{F}}T_{\tau, \eta, \bm{\delta}}(\bar{\Sigma})$.
\end{theorem}

\begin{proof}
It follows from the definitions of $T_{\tau}(\bar{\Sigma})$ and $T_{\tau, \eta, \bm{\delta}}(\bar{\Sigma})$ that $\bar{U}_{2}\subseteq\bar{U}_{1}$. Let $(x_{1}, \tilde{q}_{1})\in\mathcal{F}$ with $x_{1}\in\bar{X}_{1}$ and $\tilde{q}_{1}\in\bar{X}_{2}$, then we have that $x_{1}\in\tilde{q}_{1}$. For each $u\in\bar{U}_{2}(\tilde{q}_{1})$, we obtain that $u\in\bar{U}_{2}(\tilde{q}_{1})\subseteq \bar{U}_{2}\subseteq\bar{U}_{1}$. From the definition of $\bar{U}_{2}(\tilde{q}_{1})$, it follows that $\bar{\Delta}_{2}(\tilde{q}_{1}, u)\neq\varnothing$. If $\bar{\Delta}_{1}(x_{1}, u)=\varnothing$, then we have that $u\notin\bar{U}_{1}$, which is a contradiction. As a result, $\bar{\Delta}_{1}(x_{1}, u)\neq\varnothing$ and $u\in\bar{U}_{1}(x_{1})$, which in turn implies that $\bar{U}_{2}(\tilde{q}_{1})\subseteq\bar{U}_{1}(x_{1})$.

Given $\tilde{q}_{1}, \tilde{q}_{2}\in\bar{X}_{2}$ and $u\in\bar{U}_{2}(\tilde{q}_{1})$, define $x_{2}:=\bar{\Delta}_{1}(x_{1}, u)$. It follows from $(x_{1}, \tilde{q}_{1})\in\mathcal{F}$ that $x_{1}\in\tilde{q}_{1}$. In addition, it follows from \eqref{eqn-16} that $\tilde{q}_{1}=\psi_{2}(x_{1})$ with $q_{1}=\Sigma^{N}_{j=0}q_{1j}s_{j}$. From the zoom quantization, we have that $x_{1}\in\llbracket q_{1}-2\theta_{2}E, q_{1}+2\theta_{2}E\rrbracket$, which implies that $\|x_{1}-q_{1}\|\leq2\theta_{2}$. If $\bar{\Delta}_{1}(x_{1}, u)\cap\tilde{q}_{2}\neq\varnothing$, then there exists $\mathbf{x}(\tau, x_{1}, u)\in\tilde{q}_{2}$. Since $f$ in $\bar{\Sigma}$ satisfies the Lipschitz condition, one has that $\|\mathbf{x}(\tau, x_{1}, u)-\mathbf{x}(\tau, q_{1}, u)\|\leq e^{L_{2}\tau}\|x_{1}-q_{1}\|\leq2\theta_{2}e^{L_{2}\tau}$, which further indicates that $\tilde{q}_{2}\cap(\mathbf{x}(\tau, q_{1}, u)+2\theta_{2}e^{L_{2}\tau}\llbracket-E, E\rrbracket)\neq\varnothing$. It follows from the definition of the symbolic model $T_{\tau, \eta, \bm{\delta}}(\bar{\Sigma})$ that $\tilde{q}_{2}\in\bar{X}_{2}$, which in turn implies that $(x_{2}, \tilde{q}_{2})\in\mathcal{F}$.
\hfill$\blacksquare$
\end{proof}

Note that the overapproximation technique is applied in the transition relation \eqref{eqn-17}. Since the value $q_{1}=\Sigma^{N+1}_{j=0}q_{1j}s_{j}$ is used to approximate the time-delay state, a possible case is that $q_{1}$ does no belong to the region $\tilde{q}_{1}$. In this case, in order to measure the distance between $x_{1}$ and $q_{1}$, we overapproximate the growth bound in \eqref{eqn-17}, which thus provides more transitions in $T_{\tau, \eta, \bm{\delta}}(\bar{\Sigma})$.

\begin{remark}
\label{rmk-2}
In Theorems \ref{thm-1} and \ref{thm-2}, the feedback refinement relation is established for both the delay-free and time-delay cases. However, the techniques applied in the delay-free and time-delay cases are different. In the delay-free case, only logarithmic quantization is implemented to approximate the state set and to determine the transition relation in the symbolic model. In the time-delay case, the combination of logarithmic and zoom quantization is applied, which cannot determine the abstract sets and just lays the foundation for the approximation of the functional spaces. Besides the combination of logarithmic and zoom quantization, the spline functions are involved to determine the abstract states and the transition relation in the time-delay case. As a result, the techniques applied in the time-delay case are more complex, and lead to a refinement for the approximation in the delay-free case; see Remark \ref{rmk-1}.
\hfill $\square$
\end{remark}

\section{Numerical Example}
\label{sec-example}

In this section, a numerical example is presented to illustrate the obtained results. Consider a robotic arm that follows the Lagrange dynamics \citep{Lewis1998neural}:
\begin{equation}
\label{eqn-18}
M(\bm{p})\ddot{\bm{p}}+C(\bm{p}, \dot{\bm{p}})\dot{\bm{p}}+g(\bm{p})=W.
\end{equation}
Let $M(\bm{p})=1$, $C(\bm{p}, \dot{\bm{p}})=l$, $g(\bm{p})=m\sin(\bm{p})$, $W=u$ and define $x_{1}=\bm{p}$, $x_{2}=\dot{\bm{p}}$, then we have (see \citep{Pola2008approximately})
\begin{align}
\label{eqn-19}
\Sigma: \dot{x}_{1}=x_{2}, \quad \dot{x}_{2}=-m\sin(x_{1})-lx_{2}+u,
\end{align}
where $x_{1}$ and $x_{2}$ are respectively the angular position and rotational velocity, and $u$ is the torque which is treated as the control variable. Let $m=1.96$ and $l=1.5$. Assume that the state set is $X=[-1, 1]\times[-1, 1]$, and the input set is $U=[-2.5, 2.5]$. The applied logarithmic quantizer is given by
\begin{align}
\label{eqn-20}
Q_{1}(z):=\left\{\begin{aligned}
&\frac{(1+\eta)^{k+1}a}{(1-\eta)^{k}},  & & \frac{(1+\eta)^{k}a}{(1-\eta)^{k}}<z\leq\frac{(1+\eta)^{k+1}a}{(1-\eta)^{k+1}}; \\
&0, & & 0\leq z\leq a;  \\
& -Q(-z), & & z<0.
\end{aligned}\right.
\end{align}
To compare with the existing work \citep{Pola2008approximately}, we choose $\eta=0.2$, $a=0.4$ and $\tau=0.2$. By computation, we have that $L_{1}=6$.

According to the approximation approach in Section \ref{sec-nondelaysymbolic}, the resulting transition system is $T_{0.2, 0.2}(\Sigma)=(X_{2}, X^{0}_{2}, U_{2}, \Delta_{2}, Y_{2})$ with: (i) $X_{2}=\{\hat{q}: q=(q_{1}, q_{2}), q_{1}, q_{2}\in\{-0.72, -0.48, 0, 0.48, 0.72\}\}$; (ii) $X^{0}_{2}=X_{2}$; (iii) $U_{2}=[U]_{0.2}$; (iv) the transition relation $\Delta_{2}$ is depicted in Figure \ref{fig-2}; (v) $Y_{2}=X_{2}$; (vi) $H_{2}=\Id_{X_{2}}$. Comparing with the uniform quantization based abstraction in \citep{Pola2008approximately}, there are more (loop) transitions in $T_{0.2, 0.2}(\Sigma)$ emanating from the abstract states, which further implies that some complexity issues can be avoided; see \citep{Reissig2014feedback}. For instance, if the state 21 in Figure \ref{fig-2} is an equilibrium and is in the target set, then using the abstraction in \citep{Pola2008approximately}, no loop transition emanating from the equilibrium leads back to the equilibrium, which may result in complexity issues in terms of the refinement of abstract controllers. However, there are many loop transitions in Figure \ref{fig-2} such that the transition starting from the equilibrium will go back to the equilibrium, thereby resolving the above refinement complexity issue. Obviously, the shortest loop transition is the transition from the state 21 to the state 21.

\begin{figure}[!t]
\begin{center}
\begin{picture}(80, 90)
\put(-50,-20){\resizebox{65mm}{40mm}{\includegraphics[width=2.5in]{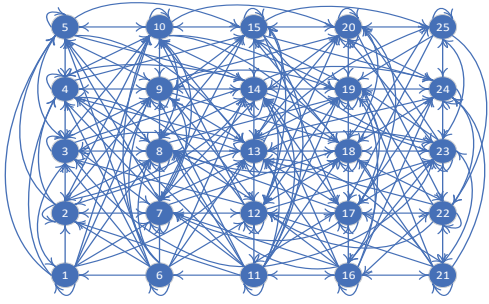}}}
\end{picture}
\end{center}
\caption{Symbolic model $T_{0.2, 0.2}(\Sigma)$ for the system $\Sigma$, where the abstract state $\hat{q}$ in $T_{0.2, 0.2}(\Sigma)$ with $q=(-0.72, -0.72)$ corresponds to the state $1$ in this figure.}
\label{fig-2}
\end{figure}

\begin{figure}[!t]
\begin{center}
\begin{picture}(80, 90)
\put(-50,-20){\resizebox{65mm}{40mm}{\includegraphics[width=2.5in]{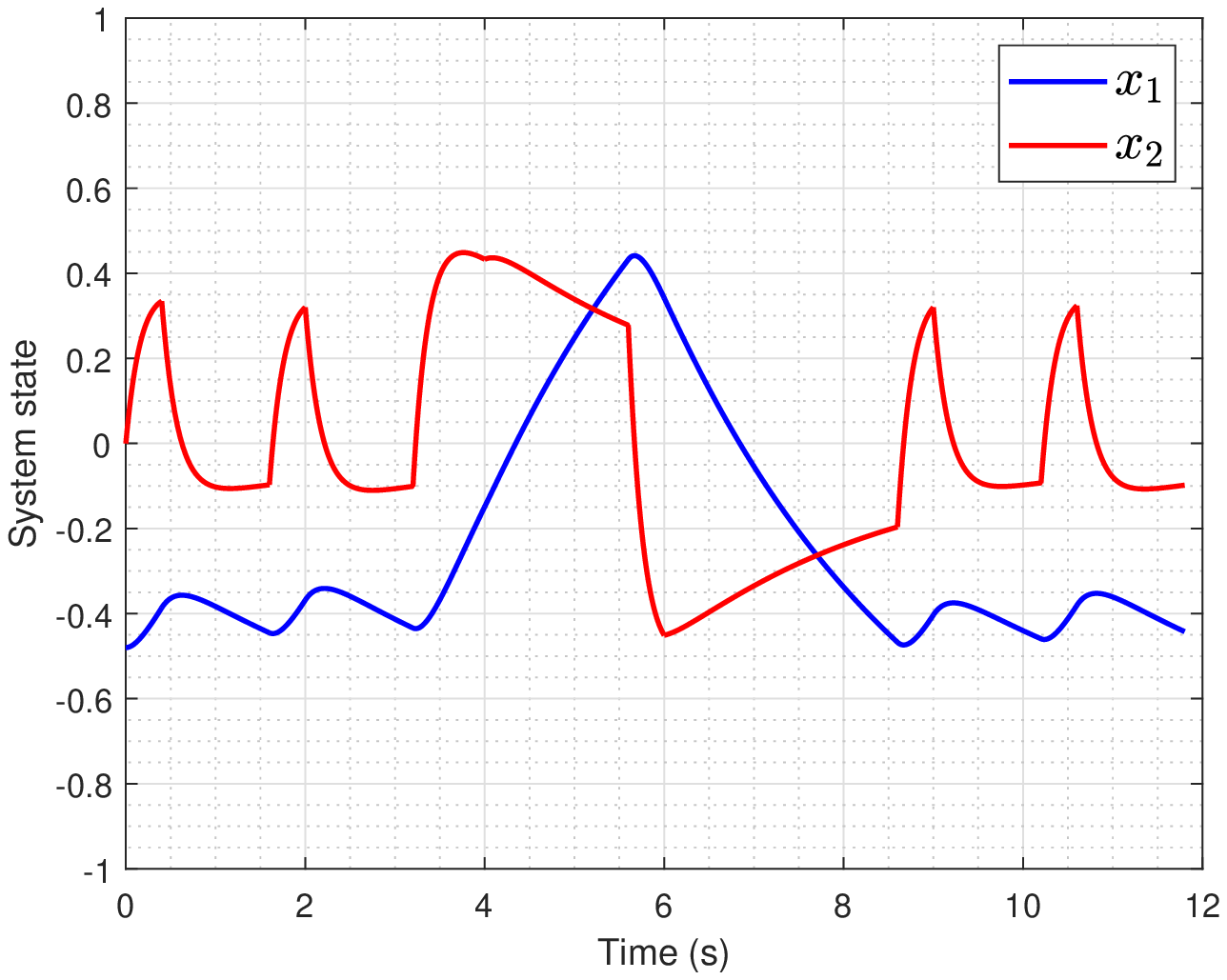}}}
\end{picture}
\end{center}
\caption{Trajectory of the control system $\Sigma$ with initial condition $(-0.48, 0)$ and control strategy synthesized on $T_{0.2, 0.2}(\Sigma)$.}
\label{fig-3}
\end{figure}

\begin{figure}[!t]
\begin{center}
\begin{picture}(80, 90)
\put(-50,-20){\resizebox{65mm}{40mm}{\includegraphics[width=2.5in]{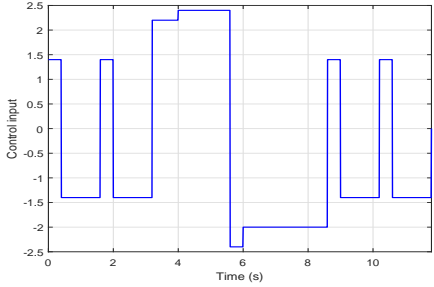}}}
\end{picture}
\end{center}
\caption{The synthesized control strategy for the system $\Sigma$ based on the constructed symbolic model $T_{0.2, 0.2}(\Sigma)$.}
\label{fig-4}
\end{figure}

In the following, the controller synthesis is illustrated for the symbolic model $T_{0.2, 0.2}(\Sigma)$. Assume that the objective is to design a controller to enforce an alternation between two different periodic motions: the first periodic motion $\mathcal{S}_{1}$ requires the state of $\Sigma$ to cycle between $(-\phi, 0)$ and $(0, 0)$, whereas the second periodic motion $\mathcal{S}_{2}$ requires the state to cycle between $(-\phi, 0)$ and $(\phi, 0)$, where $\phi$ is set as $0.48$. To achieve this objective, our control aim is to design a controller such that the system $\Sigma$ satisfies a specification $\mathcal{S}$, which requires the execution of the sequence of periodic motions $\mathcal{S}_{1}, \mathcal{S}_{1}, \mathcal{S}_{2}, \mathcal{S}_{1}, \mathcal{S}_{1}$. For this specification, a control strategy for periodic motions $\mathcal{S}_{1}$ and $\mathcal{S}_{2}$ can be obtained by performing a search on $T_{0.2, 0.2}(\Sigma)$ or by using standard methods in the context of supervisory control \citep{Ramadge1987supervisory}. A solution for the execution of $\mathcal{S}_{1}$ is given by $(-\phi, 0)\overset{1.4}{\longrightarrow}(0, 0)\overset{-1.4}{\longrightarrow}(-\phi, 0)$, and a solution for the execution of $\mathcal{S}_{2}$ is given by $(-\phi, 0)\overset{2.2}{\longrightarrow}(0, \phi)\overset{2.4}{\longrightarrow}(\phi, 0)\overset{-2.4}{\longrightarrow}(0, -\phi)\overset{-2}{\longrightarrow}(-\phi, 0)$. By combining these two solutions, a control strategy for the specification $\mathcal{S}$ is derived, and we have the following transitions: $(-\phi, 0)\overset{1.4}{\longrightarrow}(0, 0)\overset{-1.4}{\longrightarrow}(-\phi, 0)\overset{1.4}{\longrightarrow}(0, 0)\overset{-1.4}{\longrightarrow}(-\phi, 0)\overset{2.2}{\longrightarrow}(0, \phi)\overset{2.4}{\longrightarrow}(\phi, 0)\overset{-2.4}{\longrightarrow}(0, -\phi)\overset{-2}{\longrightarrow}(-\phi, 0)\overset{1.4}{\longrightarrow}(0, 0)\overset{-1.4}{\longrightarrow}(-\phi, 0)$. Note that some transitions are not obtained within one sampling period. This means that the abstract state may stay the same after certain transitions, which results from the constructed symbolic abstraction via logarithmic quantization; see also the loop transitions in Figure \ref{fig-2}. The evolution of the system state is shown in Figure \ref{fig-3}, and the control strategy is presented in Figure \ref{fig-4}. The completion time of the specification $\mathcal{S}$ is 11.8s, whereas the completion time in \citep{Pola2008approximately} is 24s, which implies the reduction of the computation time since the refinement complexity issue is resolved in this paper.

\begin{figure}[!t]
\begin{center}
\begin{picture}(80, 90)
\put(-20,-20){\resizebox{40mm}{40mm}{\includegraphics[width=2.5in]{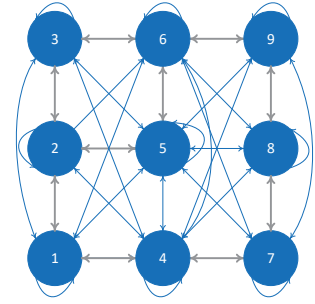}}}
\end{picture}
\end{center}
\caption{The partition of the state 13 in Figure \ref{fig-2} with zoom quantization and $\delta=0.3$.}
\label{fig-5}
\end{figure}

\begin{figure}[!t]
\begin{center}
\begin{picture}(80, 90)
\put(-50,-20){\resizebox{65mm}{40mm}{\includegraphics[width=2.5in]{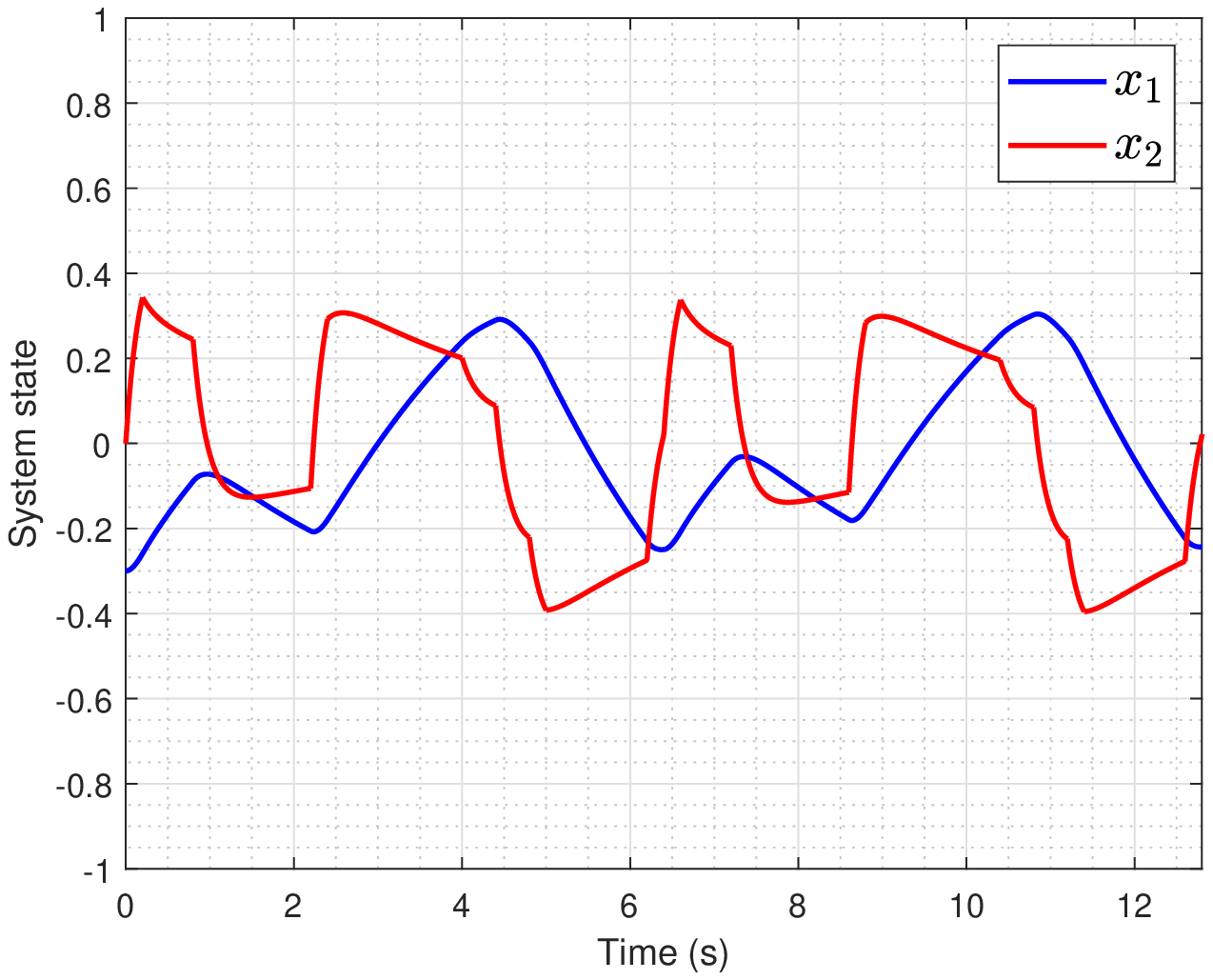}}}
\end{picture}
\end{center}
\caption{Trajectory of the control system $\Sigma$ with initial condition $(-0.4, 0)$ and control strategy synthesized on $T_{0.2, 0.2, 0.3}(\Sigma)$.}
\label{fig-6}
\end{figure}

\begin{figure}[!t]
\begin{center}
\begin{picture}(80, 90)
\put(-50,-20){\resizebox{65mm}{40mm}{\includegraphics[width=2.5in]{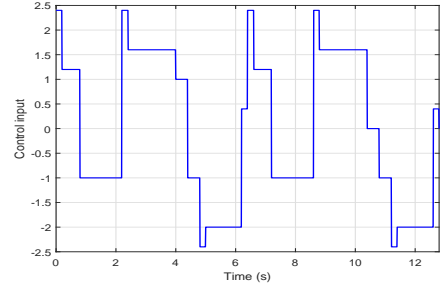}}}
\end{picture}
\end{center}
\caption{Trajectory of the control system $\Sigma$ with initial condition $(-0.48, 0)$ and control strategy synthesized on $T_{0.2, 0.2, 0.3}(\Sigma)$.}
\label{fig-7}
\end{figure}

On the other hand, if $\phi=0.3$ in the specifications, then we need to further approximate the state 13 in Figure \ref{fig-2}, which is based on Subsection \ref{subsec-delayapprox} and is presented in Figure \ref{fig-5}. The state 13 in Figure \ref{fig-2} is partitioned into 9 states in Figure \ref{fig-5}, where the transitions among these 9 states are presented. With the abstraction refinement, we can similarly obtain a symbolic model $T_{0.2, 0.2, 0.3}(\Sigma)$. A control strategy for the alternation between the specifications $\mathcal{S}_{1}$ and $\mathcal{S}_{2}$ is given by $(-\phi, 0)\overset{2.4}{\longrightarrow}(-\phi, \phi)\overset{1.2}{\longrightarrow}(0, \phi)\overset{-1}{\longrightarrow}(0, 0)\overset{1}{\longrightarrow}(-\phi, 0)\overset{2.4}{\longrightarrow}(-\phi, \phi)\overset{1.6}{\longrightarrow}(0, \phi)\overset{1}{\longrightarrow}(\phi, \phi)\overset{-1}{\longrightarrow}(\phi, 0)\overset{-2.4}{\longrightarrow}(\phi, -\phi)\overset{-2}{\longrightarrow}(0, -\phi)\overset{-2}{\longrightarrow}(-\phi, -\phi)\overset{0.4}{\longrightarrow}(-\phi, 0)$; see the gray transitions in Figure \ref{fig-5}. The control strategy is presented in Figure \ref{fig-7}, and the evolution of the system state is shown in Figure \ref{fig-6}. Comparing with the approaches in \citep{Reissig2017feedback, Pola2008approximately}, we do not need to rediscretize the state and input sets, and do not increase the number of the abstract states greatly. For instance, to achieve the aforementioned specification, 33 abstract states are involved here, whereas the rediscretization of the state set is needed and 49 abstract states are involved by using the approaches in \citep{Reissig2017feedback, Pola2008approximately}.

Assume that $W$ in \eqref{eqn-18} is of the form $\alpha x_{2t}+u(t-r)$ with bounded constants $\alpha\in\mathbb{R}$ and $\Theta, r>0$, and \eqref{eqn-18} is rewritten as a time-delay control system:
\begin{align}
\label{eqn-21}
\bar{\Sigma}: \dot{x}_{1}&=x_{2}, \quad \dot{x}_{2}=-m\sin(x_{1})-lx_{2}+\alpha x_{2t}+u(t-r).
\end{align}
For the system $\bar{\Sigma}$, we can use the zoom quantizer to further partition the abstract state obtained by the logarithmic quantization. For instance, using the quantizer \eqref{eqn-7} with $\Lambda\delta=0.1$, the abstract state 1 in Figure \ref{fig-2} is partitioned into 25 smaller regions in Figure \ref{fig-8}. For the time-delay initial trajectory (the red curve in Figure \ref{fig-8}), we only need three smaller regions (i.e., the cyan regions) to cover it. Similar to the delay-free case, we can obtain the symbolic model as proposed in Subsection \ref{subsec-delaysymbolic}, and study the control synthesis problem for a given specification.

\begin{figure}[!t]
\begin{center}
\begin{picture}(80, 90)
\put(-30,-22){\resizebox{40mm}{40mm}{\includegraphics[width=2.5in]{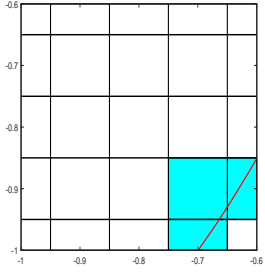}}}
\end{picture}
\end{center}
\caption{The partition of the abstract state 1 in Figure \ref{fig-2} via the zoom quantization. Only three smaller regions (the cyan part) are needed to cover the time-delay initial trajectory (the red curve).}
\label{fig-8}
\end{figure}

\section{Conclusion}
\label{sec-conclusion}

In this paper, we studied symbolic abstraction of nonlinear control systems in both the delay-free and time-delay cases. For the delay-free case, both the state and input sets were approximated via the logarithmic quantizer to reduce computational complexity, and a symbolic model was developed. For the time-delay case, the combination of the logarithmic and zoom quantizers was applied to approximate the state and input sets, and symbolic model was also constructed. In both cases, a feedback refinement relation was verified for the symbolic model and the original system. Future work will be directed to the construction of symbolic abstractions for switched control systems, and to applications of the symbolic models to multi-agent systems.

\bibliographystyle{elsarticle-num-names}

\end{document}